\documentclass[conference]{IEEEtran}
\IEEEoverridecommandlockouts
\usepackage{cite}
\usepackage{amsmath,amssymb,amsfonts,amsthm}
\usepackage{algorithmic}
\usepackage{graphicx}
\usepackage{textcomp}
\usepackage{xcolor}
\usepackage{tikz}
\usepackage{cleveref}
\usepackage{subcaption}
\usepackage{enumitem}
\usepackage{calc}

\usetikzlibrary{calc}
\usetikzlibrary{shapes}
\usetikzlibrary{positioning}
\usetikzlibrary{decorations.pathreplacing}

\newtheorem{example}{Example}
\crefname{example}{Ex.}{Examples}
\Crefname{example}{Ex.}{Examples}

\newtheorem{definition}{Definition}
\crefname{definition}{Def.}{Definitions}
\Crefname{definition}{Def.}{Definitions}

\newtheorem{theorem}{Theorem}
\crefname{theorem}{Thm.}{Theorems}
\Crefname{theorem}{Thm.}{Theorems}

\crefname{figure}{Fig.}{Figures}
\Crefname{figure}{Fig.}{Figures}

\crefname{equation}{Eq.}{Equations}
\Crefname{equation}{Eq.}{Equations}

\crefname{section}{Sec.}{Sections}
\Crefname{section}{Sec.}{Sections}

\begin{document}

\title{PERIDOT Codes: Replacing Identifiers, Sequence Numbers and Nonces with Permutations}

\author{\IEEEauthorblockN{Florian Euchner and Christian Senger}
\IEEEauthorblockA{
Institute of Telecommunications, University of Stuttgart, Stuttgart, Germany \\
Email: \{euchner, senger\}@inue.uni-stuttgart.de}
}

\maketitle

\begin{abstract}
    Identifiers and sequence numbers make up a large part of the protocol overhead in certain low-power wide-area networks.
    The requirement for cryptographic nonces in authentication and encryption schemes often demands excessively long sequence numbers, which leads to an increase in energy consumption per transmitted packet.
    In this paper, the novel PERIDOT coding scheme is proposed.
    It replaces identifiers and sequence numbers with a code, based on which receivers can identify transmitters with high confidence.
    PERIDOT is based on specially constructed integer permutations assigned to transmitters.
    An upper bound on the performance of PERIDOT codes is provided and methods for constructing particularly suitable permutations are presented.
    In practice, PERIDOT can significantly increase intervals between nonce reuses and, at the same time, reduce power consumption.
\end{abstract}


\section{Introduction}
\label{sec:intro}
Low-power wide-area network (LPWAN) protocols like LoRaWAN, Sigfox, MIOTY and others are commonly used to wirelessly transmit packets with very small application payloads, sometimes as small as a single bit.
Devices connected to such networks are often uplink-only and their transmissions are powered by small batteries or by energy harvesting.

Nevertheless, since LPWAN devices may transmit sensitive information, the various protocols define authentication schemes based on message authentication codes (MACs) and symmetric encryption schemes.
These authentication and encryption methods need nonces, i.e. a number has to be transmitted alongside the application payload and this number should be reused as rarely as possible, preferably only once.
In practice, the sequence number (SN) of packets, a counter that is traditionally used to detect packet loss, gets a dual-use as nonce.
Protocol designers have to find a trade-off between longer SNs (improved security due to longer intervals between nonce reuses) and low energy consumption:
The transmission of every single bit consumes energy and the device's overall energy consumption is almost proportional to the number of bits transmitted, since most other tasks (e.g. reading sensor values) are negligible in their energy consumption.

For instance, in the case of Sigfox, an uplink packet with a single-bit application payload has a total length of 112 bit including the preamble, 12 bit of which are used as SN and nonce \cite{sigfoxspec}.
This obviously makes replay attacks feasible since SNs have to be reused every $2^{12} = 4096$ packets.
LoRaWAN, on the other hand, employs a 32 bit SN internally, only the 16 least significant bits of which are transmitted - the receiver has to infer the remaining bits from observed uplinks \cite{loraspec}.
This scheme is stateful and potentially error-prone.

In our protocol design approach, we replace SNs altogether and merge them with identifiers into a single number, called code number (CN).
The same CN is guaranteed to recur as rarely as possible, hence it can have a dual-use as nonce.
The gist of the scheme, which we will call Permutation-based Identification and Order Tracking (PERIDOT), is best explained by two examples based on the following scenario:

\begin{figure}
    \centering
    \begin{tikzpicture}
        \node [rectangle, draw, thick, minimum height = 0.6cm, minimum width = 4cm] (tx) at (0, 0.6) {Transmitter};
        
        \node (packet) [minimum width = 4cm, minimum height = 0.5cm, draw, thick] at (2.5, -0.3) {};
        \node (nonce) [minimum width = 1cm, minimum height = 0.5cm, draw, thick, fill = black!10!white] at (1, -0.3) {\footnotesize CN};
        \node (data) [minimum width = 2.0cm, minimum height = 0.5cm, draw, thick, fill = black!10!white] at (2.5, -0.3) {\footnotesize Data};
        \node (mac) [minimum width = 1cm, minimum height = 0.5cm, draw, thick, fill = black!10!white] at (4, -0.3) {\footnotesize MAC};
        
        \node (c_pec) [minimum width = 4cm, minimum height = 0.6cm, draw, thick] at (0, -1.2) {Packet Erasure Channel};
        \node (erasure) at (3.5, -1.2) {Erasure};
        \node [rectangle, draw, thick, minimum height = 0.6cm, minimum width = 4cm] (rx) at (0, -2.6) {Receiver};

        \draw [-latex] (tx) -- (c_pec.north);
        \draw [-latex] (c_pec.south) -- (rx) node[midway, right] {$1 - \varepsilon$};
        \draw [-latex] (c_pec.east) -- (erasure) node[midway, above] {$\varepsilon$};
    \end{tikzpicture}
    \caption{Channel model: a single Packet Erasure Channel}
    \label{fig:channelmodel}
\end{figure}
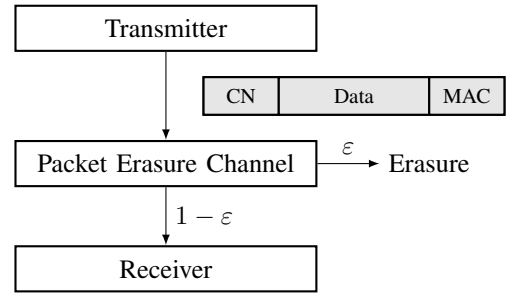

We assume that, as shown in \cref{fig:channelmodel}, a packet is transmitted over a packet erasure channel (PEC) with erasure probability $\varepsilon$.
In our model, the receiver is unaware of the occurrence of an erasure.
Let the CN be from the alphabet $Z_\mathrm{q} = \{ 0, \ldots, q - 1 \}$.
Our hard requirement is that, since CNs are used as nonces, once a particular CN has been transmitted, the following $q - 1$ packets must not contain the same CN again.
Thus, the order in which CNs are transmitted has to be a cyclic permutation of the elements of $Z_\mathrm{q}$.
That is, the permutation has only a single cycle whose orbit is $Z_\mathrm{q}$.
This permutation $\sigma: Z_\mathrm{q} \to Z_\mathrm{q}$ maps the previously transmitted CN $u$ to the next CN $v = \sigma[u]$.
We use $\circ$ to denote composition of permutations and $\sigma^\beta$ for $\beta$-fold composition.
In both examples, it is assumed that there is \emph{only a single transmitter} communicating with a single receiver.

\begin{figure*}
    \begin{subfigure}[b]{0.33\textwidth}
        \centering
        \includegraphics{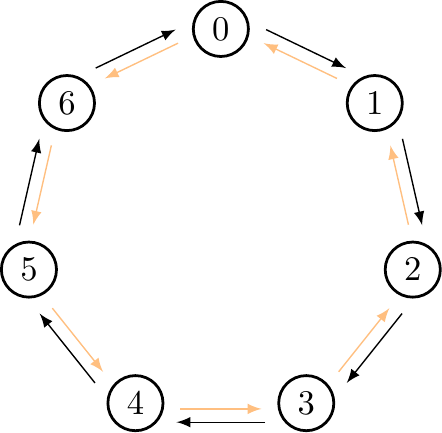}
        \caption{Permutations in a $(7, 3)$-proper set}
        \label{fig:p73}
    \end{subfigure}
    \begin{subfigure}[b]{0.33\textwidth}
        \centering
        \includegraphics{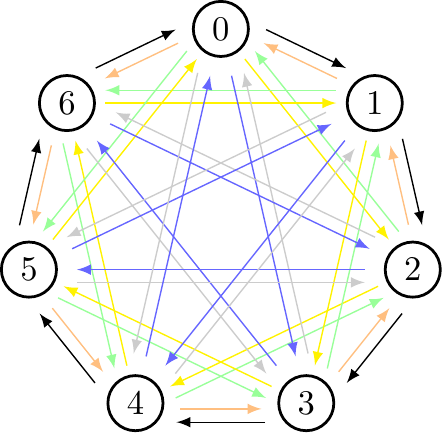}
        \caption{Permutations in a $(7, 1)$-proper set}
        \label{fig:p71}
    \end{subfigure}
    \begin{subfigure}[b]{0.33\textwidth}
        \centering
        \includegraphics{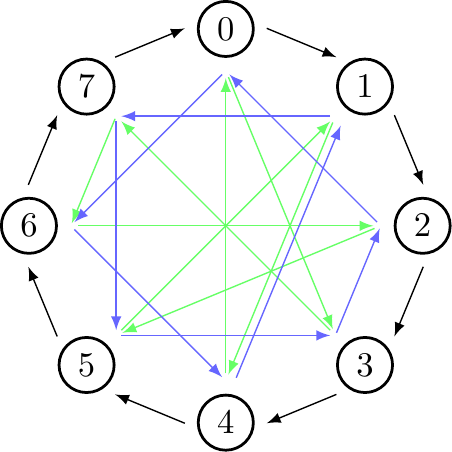}
        \caption{Permutations in an $(8, 2)$-proper set}
        \label{fig:p28}
    \end{subfigure}
    \caption{Illustrations of different permutations (order in which CNs are sent), one color per permutation}
\end{figure*}

\begin{example}
    For some $q \geq 3$, define an incrementing and a decrementing permutation as
    \begin{equation*}
        \sigma_\mathrm{inc}: u \mapsto (u + 1) \bmod q \qquad \sigma_\mathrm{dec}: u \mapsto (u - 1) \bmod q
    \end{equation*}
    
    and assign the permutations to one of two transmitters each.
    The permutations are illustrated in \cref{fig:p73} for $q = 7$.
    One of the two transmitters is selected to send a packet containing an arbitrary CN $u \in Z_\mathrm{q}$ and then a packet with CN $v = \sigma[u]$, where $\sigma$ is its assigned permutation.
    The receiver receives the packets in that order.
    For erasure probability $\varepsilon = 0$, the receiver is always able to identify the device based on $u$ and $v$ (e.g. $u = 5$ and $v = 6$ can only stem from permutation $\sigma_\mathrm{inc}$).
    Even for nonzero $\varepsilon$, the receiver can still try to identify the transmitter using two non-erased packets based on which permutation, $\sigma_\mathrm{inc}$ or $\sigma_\mathrm{dec}$, the CNs contained in the packets were more likely drawn from.

    \label{ex:inc_dec}
\end{example}

\begin{example}
    Suppose $\varepsilon = 0$, i.e. there are no erasures.
    If $q$ is prime and there are $q - 1$ transmitters or less, we can uniquely assign a number $\Delta \in Z_\mathrm{q} \backslash \{ 0 \}$, called increment, to each of them, so that by knowing $\Delta$, the transmitter is identified.
    The transmitters are configured to send nonces according to the permutation $\omega_\Delta: u \mapsto (u + \Delta) \bmod q$.
    Then the receiver can determine a transmitter's identity based on two consecutively received nonces $u$ and $v$: $\Delta = (v - u) \bmod q$.
    This can also be seen from the illustration for $q = 7$ in \cref{fig:p71}.

    For non-prime $q$, $\omega_\Delta: u \mapsto (u + \Delta) \bmod q$ is no longer a permutation of $Z_\mathrm{q}$ in case $\gcd(\Delta, q) \neq 1$.
    For most $q$, solutions with $q - 1$ permutations can nevertheless be found: Formalizing the requirements, we want to find a set of permutations $P$ such that
    \begin{equation*}
        \forall (u, v) \in Z_\mathrm{q}^2 ~ \exists_{\leq 1} \sigma \in P: ~ v = \sigma[u].
    \end{equation*}
    In other words, the transition from $u$ to $v$ must be unique among all transmitters.
    Regarding $Z_\mathrm{q}$ as the vertex set of a directed graph, an equivalent problem is to find a set of edge-disjoint directed Hamiltonian cycles in the complete directed graph $K^*_q$, where the permutations $\sigma \in P$ are given by the order in which cycles visit nodes.
    From Walecki's construction \cite{walecki} and from Tillson's paper \cite{tillson}, we find solutions with $q - 1$ directed Hamiltonian cycles and thus permutations for all $q \in Z_\mathrm{q} \backslash \{ 1, 2, 4, 6 \}$.
    The cases $q \in \{ 1, 2 \}$ are trivial and, by full search, we find that it is not possible to find $q - 1$ directed Hamiltonian cycles for $q \in \{ 4, 6 \}$.

    Other than in \cref{ex:inc_dec}, any packet loss between the reception of $u$ and $v$ results in misidentification.

    \label{ex:hamiltonian}
\end{example}

These examples show that CNs may carry both identifying information and information facilitating the detection of lost packets.
In \cref{ex:inc_dec}, the receiver can only discern two devices, but is guaranteed to tolerate up to $\lfloor (q - 1) / 2 \rfloor - 1$ consecutively lost packets whereas in \cref{ex:hamiltonian}, packet loss is not tolerable, but the receiver can identify up to $q - 1$ devices.
These examples are two extreme ends in a spectrum of conceivable schemes that permit distinguishing transmitters based on CN permutations assigned to them.
In the following \cref{sec:lqsets}, we will formalize this notion, define a necessary and sufficient criterion for such schemes and provide an upper bound.
In \cref{sec:construction}, we will provide a particularly suitable construction of sets of permutations and in \cref{sec:multitx} and \cref{sec:sigfox} we will explore how PERIDOT can be practically implemented for systems with many transmitters.

\section{$(q, l)$-Proper Sets and Upper Bound}
\label{sec:lqsets}
We assume that the receiver detects two packets from the transmitter with CNs $u, v \in Z_\mathrm{q}$ in that order.
We count the number of packets that were transmitted after $u$, including the one containing $v$, and call this number $\beta$.
The number of packets lost between $u$ and $v$ is thus $\beta - 1$ and we always have $\beta \geq 1$.
Moreover, we assume that at most $l - 1$ packets are lost over the PEC, such that $\beta$ can be bounded by $\beta \leq l$.

\begin{definition}[$(q, l)$-proper sets]
    A set $P$ of cyclic permutations with orbit $Z_\mathrm{q}$ is called a $(q, l)$-proper set for $l \in \mathbb N_0, l < q$ if, and only if, we have
    \begin{equation*}
        \forall (u, v) \in Z_\mathrm{q}^2 ~ \forall \beta \in \mathbb N, \beta \leq l ~ \exists_{\leq 1} \sigma \in P: ~ v = \sigma^\beta[u].
    \end{equation*}
    \label{def:plq}
\end{definition}

Intuitively, \cref{def:plq} states that a $(q, l)$-proper set $P$ is a set in which we can \emph{uniquely} identify each permutation $\sigma \in P$ based on an ordered pair $(u, v)$, where this ordered pair is generated by arbitrarily choosing $u \in Z_\mathrm{q}$ and then computing $v$ as $v = \sigma^\beta[u]$ with some $\beta \leq l$.
This corresponds to a device successfully transmitting the CN $u$, then loosing $\beta - 1 \leq l - 1$ messages (with CNs $\sigma[u]$, $\sigma^2[u]$, \ldots, $\sigma^{\beta - 1}[u]$) and then successfully transmitting $v = \sigma^\beta[u]$.
The receiver can identify the permutation $\sigma$ and thus the transmitter based on $(u, v)$.

We see that the permutations in a $(q, l)$-proper set are suitable sequences of CNs.
In this notation, \cref{ex:inc_dec} provides a construction for a $(q, \lfloor (q - 1) / 2 \rfloor)$-proper set whereas \cref{ex:hamiltonian} describes $(q, 1)$-proper sets.
The following is a nontrivial example for a $(q, l)$-proper set (for an unrealistically small alphabet size of $q = 8$):

\begin{figure}
    \centering
    \includegraphics{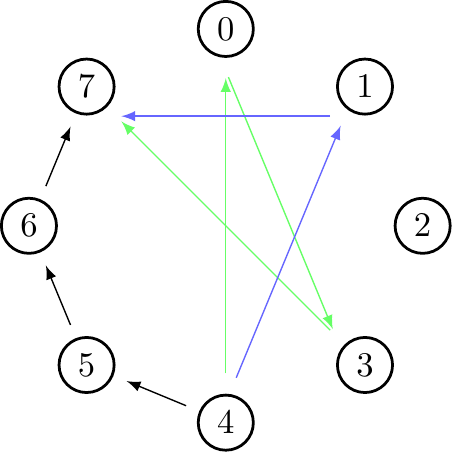}
    \caption{Transitions from $u = 4$ to $v = 7$ in a $(8, 2)$-proper set}
    \label{fig:example_transition}
\end{figure}

\begin{example}
    A $(8, 2)$-proper set is given by $P = \{ \sigma, \textcolor{green!50!black}{\rho}, \textcolor{blue!50!black}{\pi} \}$ with
    \begin{align*}
        \textcolor{black} \sigma &= (0, 1, 2, 3, 4, 5, 6, 7) \qquad \textcolor{green!50!black} \rho = (0, 3, 7, 6, 2, 5, 1, 4) \\
        \textcolor{blue!50!black} \pi &= (0, 6, 4, 1, 7, 5, 3, 2)
    \end{align*}

    $\sigma$, $\rho$ and $\pi$ are given in cycle notation.
    There is no $(8, 2)$-proper set containing more than $|P| = 3$ permutations.
    $P$ was found by computer search and is illustrated in \cref{fig:p28}.
    For this exemplary set of permutations, if the receiver observes $u = 4$ and $v = 7$ in that order, it would generate the transition paths that $\sigma$, $\rho$ and $\pi$ would have to take to get from $u$ to $v$ as drawn in \cref{fig:example_transition}.
    From the fact that the path for permutation $\pi$ is shortest, it concludes that the most likely scenario is that the device to which $\pi$ was assigned transmitted $u$ and $v$, and that one packet was lost in between (the one with a CN of $1$).
\end{example}

The previous examples suggest that for fixed $l$ and $q$, there appears to be a maximum number of permutations that fit into a $(q, l)$-proper set, i.e. a maximum value for $m = |P|$.
This has important practical ramifications:
For a an alphabet size $q$, which is predetermined by the number of bits allocated to a CN, and $l - 1$ tolerable packet losses, we can only assign $m$ permutations to transmitters and thus only identify up to $m$ devices using PERIDOT.
Therefore, we now seek to formulate an upper bound that relates $m$, $l$ and $q$.
Such a bound will result as a consequence of an equivalent condition for $(q, l)$-proper sets based on what we call $l$-\textit{follower sets}:
\begin{definition}[$l$-follower sets]
    The $l$-follower set of some $u \in Z_\mathrm{q}$ in the permutation $\sigma$ of $Z_\mathrm{q}$ is given by
    \begin{equation*}
        \Omega_l^\sigma(u) = \left\{ \sigma^\beta[u] ~ | ~ 1 \leq \beta \leq l \right\}.
    \end{equation*}
\end{definition}

\begin{figure}
    \centering
    \begin{tikzpicture}
            \draw [radius = 0.28, very thick, draw = black] (0.0, 0) circle node (c0) {$7$};
            \draw [radius = 0.28, very thick, draw = black] (1.1, 0) circle node (c1) {$6$};
            \draw [radius = 0.28, very thick, draw = black] (2.2, 0) circle node (c2) {$9$};
            \draw [radius = 0.28, very thick, draw = black] (3.3, 0) circle node (c3) {$4$};
            \draw [radius = 0.28, very thick, draw = black] (4.4, 0) circle node (c4) {$5$};
            \draw [radius = 0.28, very thick, draw = black] (5.5, 0) circle node (c5) {$2$};

            \draw [-latex, shorten >= 0.1cm, dashed] ($(c0) - (1, 0)$) -- (c0);
            \draw [-latex, shorten <= 0.1cm, shorten >= 0.1cm] (c0) -- (c1);
            \draw [-latex, shorten <= 0.1cm, shorten >= 0.1cm] (c1) -- (c2);
            \draw [-latex, shorten <= 0.1cm, shorten >= 0.1cm] (c2) -- (c3);
            \draw [-latex, shorten <= 0.1cm, shorten >= 0.1cm] (c3) -- (c4);
            \draw [-latex, shorten <= 0.1cm, shorten >= 0.1cm] (c4) -- (c5);
            \draw [dashed, shorten <= 0.1cm] (c5) -- ($(c5) + (0.8, 0)$);

            \draw [color = red!50!black] (1.6, -0.5) rectangle (4.9, 0.5) node[midway] (omegamid) {};
            \node [above = 0.4cm of omegamid, color = red!50!black] {follower set $\Omega_3^\sigma(6)$};
        \end{tikzpicture}
    \caption{$3$-follower set of $6$ in some permutation}
    \label{fig:followerset}
\end{figure}
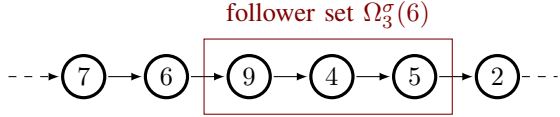

Figuratively, $\Omega_l^\sigma(u)$ contains all the CNs that succeed $u$ in the next $l$ messages, as illustrated in \cref{fig:followerset}.
It is clear that, if we have a set of alphabet permutations $P$ and all $l$ CNs that succeed $u$ are distinct from any permutation to another, then we can identify each permutation $\sigma \in P$ based on $(u, v)$ where $v$ is in the $l$-follower set of $u$ in the permutation $\sigma$.

Therefore, using the definitions of $(q, l)$-proper sets and follower sets, we can equivalently say that $P$ is a $(q, l)$-proper set with $l < q$ if, and only if, $\forall u \in Z_\mathrm{q}$ we have

\begin{equation}
    \forall \sigma, \pi \in P, ~ \sigma \neq \pi: ~ \Omega_l^\sigma(u) \cap \Omega_l^\pi(u) = \emptyset.
    \label{eq:plq_followerset}
\end{equation}

We can now derive an upper bound on $m = |P|$ from \cref{eq:plq_followerset}.
Since the sets $\Omega_l^\sigma$ are disjoint for different $\sigma \in P$ and because $|\Omega_l^\sigma(u)| = l$, we have
\begin{equation}
    \left| \bigcup_{\sigma \in P} \Omega_l^\sigma(u) \right| = \sum_{\sigma \in P} |\Omega_l^\sigma(u)| = |P| l \leq |Z_\mathrm{q} \backslash \{ u \}| = q - 1.
    \label{eq:followerset_almost_disjoint}
\end{equation}

With $m = |P|$, this means
\begin{equation*}
    lm + 1 \leq q.
\end{equation*}

Adding the restriction that $m$ must be integer and solving for $m$ finally yields the following upper bound:

\begin{theorem}
    For any $(q, l)$-proper set $P$ holds:
    \begin{equation*}
        m = |P| \leq \left\lfloor \frac{q - 1}{l} \right\rfloor
    \end{equation*}
    \label{thm:upperbound}
\end{theorem}

A similar upper bound (but for a more restrictive problem and an entirely different application) has previously been derived in \cite[eq. 1]{klove}.
On the basis of \cref{thm:upperbound}, we are now able to assess how ``good'' a $(q, l)$-proper set construction is:
We will call a $(q, l)$-proper set \emph{quasiperfect} if it reaches the upper bound, and \emph{perfect} if, additionally, $(q - 1) / l$ is integer.
We will see in the following that the bound in \cref{thm:upperbound} is tight.

\section{A Quasiperfect Construction}
\label{sec:construction}
For some specific values of $q$, $l$ that are also suitable for practical implementations of PERIDOT, it is possible to construct a quasiperfect $(q, l)$-proper set.
This construction is based on transmitters sending CNs in an order that is determined by an arithmetic progression modulo $q$, that is, a sequence of the form
\begin{equation}
    0, ~ \Delta \bmod q, ~ 2 \Delta \bmod q, ~ \ldots~.
\end{equation}
As in \cref{ex:hamiltonian}, the mapping from previous to next CN is given by the permutation $\omega_\Delta$ with \emph{increment} $\Delta \in Z_\mathrm{q} \backslash \{ 0 \}$ and
\begin{equation}
    \omega_\Delta: u \mapsto (u + \Delta) \bmod q.
    \label{eq:map}
\end{equation}
For $\omega_\Delta$ to be cyclic with orbit $Z_q$, we need to have $\gcd(\Delta, q) = 1$.
Note that modular arithmetic progressions have the added benefit of being easy to store (only need to store $q$, $\Delta$ at transmitter), which makes the implementation of PERIDOT on low-powered LPWAN end nodes feasible.

\begin{theorem}
    For prime $p$, $l \, | \, (p - 1)$ and $\omega_\Delta$ as in \cref{eq:map},
    \begin{equation*}
        P = \{ \omega_{1 + il} ~ | ~ i = 0, \ldots, p - 1, ~ i \neq (p - 1) / l \}
    \end{equation*}
    is a quasiperfect $(pl, l)$-proper set with $m = |P| = p - 1$.
    \label{thm:construction}
\end{theorem}
\begin{proof}
    We have to prove that
    \begin{enumerate}[label=(\roman*)]
        \item $P$ is a $(q, l)$-proper set with $q = pl$ (requiring the elements of $P$ to be cycles with orbit $Z_q$) and that
        \label{proof:plqset}
        \item $P$ is quasiperfect.
    \end{enumerate}

    Since we have $\gcd(1 + il, q) = 1$ for $i = 0, \ldots, p - 1, ~ i \neq (p - 1) / l$, the elements of $P$ are indeed cyclic permutations with orbit $Z_q$.
    The simple, but technical proof for this is omitted and we will instead focus on proving that $P$ is $(l, q)$-proper, which we will prove by contradiction and using \cref{def:plq}.

    Assume to the contrary that $P$ is not a $(q, l)$-proper set.
    Then, there are some $i_1, i_2 \in \{ 0, \ldots, p - 1 \} \backslash \{(p - 1) / l\}, ~ i_1 \neq i_2$ such that for some $\beta_1, \beta_2 \in \{ 1, \ldots, l \}$, we have
    \begin{equation*}
        \omega^{\beta_1}_{1 + i_1 l}[u] = \omega^{\beta_2}_{1 + i_2 l}[u].
    \end{equation*}

    By the definition of $\omega_\Delta$, we obtain \cref{eq:proof_contradiction_modq_unsimplified}, which can be simplified to \cref{eq:proof_contradiction_modq}:
    \begin{align}
        u + (1 + i_1 l) \beta_1 &\equiv u + (1 + i_2 l) \beta_2 \pmod q
        \label{eq:proof_contradiction_modq_unsimplified} \\
        \iff \qquad (1 + i_1 l) \, \beta_1 &\equiv (1 + i_2 l) \, \beta_2 \pmod q.
        \label{eq:proof_contradiction_modq}
    \end{align}
    Now since $q = pl$, the equivalence in \cref{eq:proof_contradiction_modq} must also hold modulo $l$, from which we can conclude that $\beta_1 \equiv \beta_2 \pmod l$.
    By the premise, we have $1 \leq \beta_1, \beta_2 \leq l$, so $\beta_1 = \beta_2$.
    Simplifying \cref{eq:proof_contradiction_modq} with $\beta_1 = \beta_2$, we get
    \begin{equation*}
        i_1 \equiv i_2 \pmod q
    \end{equation*}
    and, because $i_1, i_2 < p \leq q$, we conclude $i_1 = i_2$. This is a contradiction.

    $P$ is quasiperfect, since $m = |P| = p - 1$ and $q = pl$, so the upper bound from \cref{thm:upperbound} becomes
    \begin{equation*}
        p - 1 \leq \left\lfloor \frac{p l - 1}{l} \right\rfloor = \left\lfloor p - \frac{1}{l} \right\rfloor,
    \end{equation*}
    which is fulfilled with equality, implying quasiperfectness (but not perfectness in general).
\end{proof}

A construction similar to, but more restrictive than \cref{thm:construction} has been previously described in \cite[Theorem 1]{klove}, but, again, for a different application.
$(q, l)$-proper sets can also be constructed by adapting techniques applied for the construction of flash memory codes as presented in \cite{klove}, \cite{magnitudefour} and \cite{somecodes}.

\section{Application in Multi-Transmitter Systems}
\label{sec:multitx}

\tikzset{radiation/.style={{decorate, decoration = {expanding waves, angle = 90, segment length = 3.5pt}}}}
\tikzset{smallradiation/.style={{decorate, decoration = {expanding waves, angle = 90, segment length = 1.5pt}}}}
\tikzset{bsconnect/.style={thick, shorten <= -1pt, shorten >= -1pt}}

\newcommand{\basestation}{%
    \begin{tikzpicture}[scale = 0.4]
        \draw [thick] (0, 1) -- (0, 1.2) node[thick, circle, draw, fill = white, inner sep = 1pt] {};
        \draw [thick] (-0.6, -0.8) -- (0, 1)
            node[pos=0.1] (l1) {}
            node[pos=0.2] (l2) {}
            node[pos=0.3] (l3) {}
            node[pos=0.4] (l4) {}
            node[pos=0.5] (l5) {}
            node[pos=0.6] (l6) {}
            node[pos=0.7] (l7) {}
            node[pos=0.8] (l8) {}
            node[pos=0.9] (l9) {};

        \draw [thick] (0.6, -0.8) -- (0, 1)
            node[pos=0.1] (r1) {}
            node[pos=0.2] (r2) {}
            node[pos=0.3] (r3) {}
            node[pos=0.4] (r4) {}
            node[pos=0.5] (r5) {}
            node[pos=0.6] (r6) {}
            node[pos=0.7] (r7) {}
            node[pos=0.8] (r8) {}
            node[pos=0.9] (r9) {};

        \draw [thick] (0.0, -1) -- (0, 1)
            node[pos=0.1] (m1) {}
            node[pos=0.2] (m2) {}
            node[pos=0.3] (m3) {}
            node[pos=0.4] (m4) {}
            node[pos=0.5] (m5) {}
            node[pos=0.6] (m6) {}
            node[pos=0.7] (m7) {}
            node[pos=0.8] (m8) {}
            node[pos=0.9] (m9) {};

        \draw [bsconnect] (l2) -- (m2); \draw (m2) [bsconnect] -- (r2);
        \draw [bsconnect] (l3) -- (m3); \draw (m3) [bsconnect] -- (r3);
        \draw [bsconnect] (l4) -- (m4); \draw (m4) [bsconnect] -- (r4);
        \draw [bsconnect] (l5) -- (m5); \draw (m5) [bsconnect] -- (r5);
        \draw [bsconnect] (l6) -- (m6); \draw (m6) [bsconnect] -- (r6);

        \draw[thick, smallradiation, decoration={angle = 130}] (0, 1.2) -- +(90:0.5);
    \end{tikzpicture}
}

\newcommand{\transmitter}{
    \begin{tikzpicture}[scale = 0.6]
        \draw [thick, rounded corners = 0.2em] (0, 0) rectangle (0.8, 0.8);
        \draw [thick] (0.4, 0.8) -- (0.4, 1.3) node[circle, inner sep = 1pt, draw, thick, fill = white] {};
        \draw [smallradiation, decoration={angle = 130}] (0.4, 1.3) -- +(90:0.3);
    \end{tikzpicture}
}

Using the previously described PERIDOT code construction, we can now address the practical implementation of the scheme for use with LPWANs.
In our LPWAN model, we assume that end devices spontaneously emit uplink-only packets that one or multiple nearby base stations can pick up.
As shown in \cref{fig:lpwan_model}, the channels between any device and the base station are modelled as independent unidirectional packet erasure channels.
Furthermore, we assume that all packets originating from all devices globally are being processed by a central service, hereafter called \emph{backend}.
The assumption of the backend being centralized is not necessary, but will greatly simplify the subsequent explanations.
The backend has a database of devices, containing the previously received CNs $u$ for each device as well as each packet's timestamp and a very approximate device location.
We assume that if the base station picks up a packet containing the CN $v$,

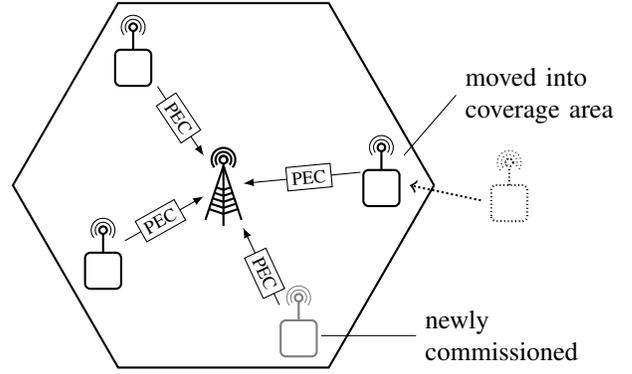
\begin{figure}
    \centering
    \begin{tikzpicture}
        \newdimen\R
        \R = 2.8cm
        \draw [thick] (0:\R) \foreach \x in {60, 120, ..., 360} {  -- (\x:\R) };
        \node (bs) [inner sep = 0pt] {\basestation};

        \node (tx_bottomleft) at (-1.6, -0.9) {\transmitter};
        \node (tx_topleft) at (-1.2, 1.8) {\transmitter};
        
        \node (movedtx_old) [densely dotted] at (3.8, 0.0) {\transmitter};
        \node (movedtx) at (2.1, 0.2) {\transmitter};
        \draw [->, densely dotted, thick] ($(movedtx_old.west) + (0, -0.2)$) -- ($(movedtx.east) + (0, -0.2)$);
        \node (movedtx_label) [align = left] at (4.2, 1.2) {moved into\\coverage area};
        \draw (movedtx_label.west) -- ($(movedtx) + (0.3, 0.2)$);
        
        \node (newtx) [color = gray] at (1.0, -1.8) {\transmitter};
        \node [align = left] (newtxlabel) at (3.7, -2.0) {newly\\commissioned};
        \draw (newtxlabel.west) -- ($(newtx) + (0.3, -0.2)$);
        
        \draw [shorten <= -0.1cm, -latex] ($(newtx) + (-0.3, 0.3)$) -- (bs) node[inner sep = 2pt, pos = 0.35, draw, rectangle, fill = white, sloped] {\scriptsize PEC};
        \draw [shorten <= -0.1cm, -latex] (movedtx) -- (bs) node[inner sep = 2pt, pos = 0.4, draw, rectangle, fill = white, sloped] {\scriptsize PEC};
        \draw [shorten <= -0.1cm, -latex] (tx_topleft) -- (bs) node[inner sep = 2pt, pos = 0.4, draw, rectangle, fill = white, sloped] {\scriptsize PEC};
        \draw [shorten <= -0.1cm, -latex] (tx_bottomleft) -- (bs) node[inner sep = 2pt, pos = 0.4, draw, rectangle, fill = white, sloped] {\scriptsize PEC};
    \end{tikzpicture}
    \caption{Model of LPWAN base station's coverage area}
    \label{fig:lpwan_model}
\end{figure}

\begin{enumerate}[label=(\Alph*)]
    \renewcommand{\theenumi}{\Alph{enumi}}
    \item it is very likely that it or a nearby base station has also picked up at least one of the previous $l$ packets from the same device, which has already been identified \label{case:known}
    \item or, otherwise, $l$ or more packets were erased by the channel, the device is newly commissioned or the device is mobile and has moved rapidly. \label{case:unknown}
\end{enumerate}

In all cases, the goal is to identify the transmitter, i.e. find out which permutation $\sigma$ the CN $v$ was drawn from.
The backend can handle all cases using a probability-based algorithm that is illustrated in \cref{fig:backend_algorithm} and which is only briefly outlined and underpinned by some examples in the scope of this paper.
It treats the vast majority of packets that are easy to identify (case \ref{case:known}) in a special way so that the few packets that cannot be identified so easily (case \ref{case:unknown}) can undergo computationally more complex treatment.

In an initial processing step for each received packet, the backend iterates over its database of known, identified devices and, for each device (hereafter called \emph{candidate}) computes a probability estimate that it was this particular device that transmitted $v$.
In many cases, candidates can be assigned a zero probability, e.g. in case a device was recently seen in a completely different geographic region.
For the candidates that remain, the backend computes the number of packet losses that must have occurred to make the observation of $v$ plausible, given that the last observed packet by the candidate contained the CN $u$ (from backend's database).
The number of packets transmitted after $u$, including $v$, is denoted by $\beta$ with $1 \leq \beta$ and we assume $\beta \leq q$.
Since the candidate device's assigned permutation $\sigma$ is known, this number has to fulfill $v = \sigma^\beta[u]$.
For a modular arithmetic progression-based $(q, l)$-proper set construction where $\sigma = \omega_\Delta$ for some $\Delta$, this simplifies to
\begin{equation}
    \beta \equiv \Delta^{-1} (v - u) \pmod q.
    \label{eq:map_tau}
\end{equation}
Now, using the PEC model assumption with erasure probability $\varepsilon$, the probability for $\beta$ packet losses is given by
\begin{equation*}
    P_\beta = (1 - \varepsilon) ~ \varepsilon^{\beta - 1}.
\end{equation*}
Intuitively, this means that if $v$ appears soon after $u$ in the permutation $\sigma$, it is reasonable to assume that $u$ and $v$ were transmitted by the same device.
The accuracy of probability estimates may be further augmented based on methods such as behavioral analyses and signal strength measurements.

\begin{figure}
    \centering

    \begin{tikzpicture}
        \node (search_matches) [draw, rectangle, align = left, inner sep = 5pt] at (0, 0) {\footnotesize Search for matching\\[-0.2em] \footnotesize identified devices};
        \node (decide_matches) [draw, diamond, align = left, aspect = 2.2] at (0, -1.6) {\footnotesize One probable\\[-0.2em] \footnotesize candidate only?};
        \node (decide_mac) [draw, diamond, aspect = 2.2] at (0, -3.6) {\footnotesize MAC is valid?};
        \node (identified) [draw, rectangle, inner sep = 5pt, minimum height = 1cm] at (0, -5.5) {\footnotesize Device identified!};
        \node (viterbi) [draw, rectangle, inner sep = 5pt, align = left, minimum height = 1cm] at (3, -5.5) {\footnotesize Treat as Hidden \\[-0.2em] \footnotesize Markov Model};

        \draw [-latex] (search_matches) -- (decide_matches);
        \draw [-latex] (decide_matches.south) -- (decide_mac.north) node [midway, anchor = west] {Yes};
        \draw [-latex] (decide_mac.south) -- (identified) node [midway, anchor = west] {Yes};

        \draw [-latex] (decide_matches.east) -| (viterbi.north) node[midway, anchor = south east] {No};
        \draw (decide_mac.east)  -| (viterbi.north) node[midway, anchor = south east] {No};
    \end{tikzpicture}

    \caption{Overview of backend identification algorithm}
    \label{fig:backend_algorithm}
\end{figure}
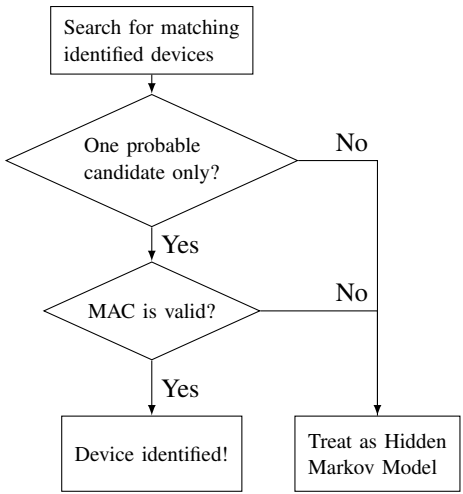

For the packets that belong to case \ref{case:known}, which we expect to be the vast majority of packets, these criteria will likely let a single device emerge as the only probable transmitter.
For this device, the MAC value is checked (i.e., the MAC gets a dual-use as both authenticity check and final identity check) and if it also matches, the packet can be attributed to the transmitter.

All other cases, that is case \ref{case:unknown} (no device matches the packet with high enough probability) or case \ref{case:known}, but multiple devices match, are treated in the same way:
the packet has to be processed by a more sophisticated algorithm that treats the entirety of all yet unidentified devices as a Hidden Markov Model (HMM).
It is for this that we constructed $(q, l)$-proper sets and assigned special permutations to devices - thanks to them, the backend will still be able to identify devices, almost always upon reception of the \emph{second} packet by the same device.

A HMM consists of a Markov chain, i.e. a set of states the system can be in, as well as a set of possible observable values.
In our case, the states of the Markov chain are given by the combined state of the \emph{entirety} of all the unidentified LPWAN end devices.
This state consists of \emph{all} the previously successfully transmitted CNs of \emph{all} unidentified devices, which is an integer from the alphabet $Z_\mathrm{q}$ for each device.
Overall, the state can thus be described by the vector $\underline u \in Z_\mathrm{q}^{d}$ for $d$ unidentified devices in the whole network, so there are a total of $q^{d}$ possible states.
In our case, an observation in the HMM is made whenever a packet is received which could not be unambiguously assigned to a device in the previous steps.

In this context, the goal of identifying the transmitter of each packet can be reformulated as guessing the most likely sequence of hidden states from the observations, since this sequence also automatically reveals which transitions must have occurred in order to produce it.
Trivially, if two state vectors differ in more than one CN, the transition probability between them is zero since packet transmissions can be assumed to occur one at a time.
Other transition probabilities between hidden states can be computed as previously explained (compute $P_\beta$, possibly augment with metadata).

The \emph{Viterbi algorithm} can be employed to find the most likely sequence of hidden states efficiently.
Due to the large number of states in the Markov chain ($q^{d}$), a practical implementation has to stop tracking unlikely paths if their probability drops below some threshold.
The way the algorithm works is best explained by example.
In order to simplify explanations, we will use unrealistically small values for $q$.

In \cref{example:known} and \cref{example:unknown}, we assume a construction as in \cref{thm:construction} with $l = 10$ and $p = 11$, so the set of increments is $\left\{ 1, 21, 31, 41, 51, 61, 71, 81, 91, 101 \right\}$ and $q = 110$.
Furthermore, we assume that identification is performed only based on CNs and an approximate (within tens to hundreds of kilometers accuracy) device geolocation.

\begin{example}
    The backend has already identified and located the asset tracker device T with increment $\Delta_\mathrm{T} = 31$ in Rotterdam.
    The last CN known to have originated from T is $u = 12$.
    Several hours later, a base station in Amsterdam, $60\,\mathrm{km}$ away, receives a packet with a CN $v = 43$, which is a plausible next CN for T.
    There are no other identified devices in all of Europe for which $v = 43$ is plausible, so we check the packet's MAC which is also valid.
    It can be assumed that the device has been successfully identified.

    The only reasonable scenario for erroneous identification is, if some other previously inactive device that was just commissioned happens to have $v = 43$ as first CN to transmit, happens to be in the same geographic area as T, and the MAC happens to be valid, which is extremely unlikely overall.
    \label{example:known}
\end{example}

With large enough values of $p$, $l$ and thereby $q$, we expect the error scenario in \cref{example:known} to hardly ever occur.

\begin{example}
    The backend has not yet identified the devices with increments $1$, $21$ and $91$.
    A base station in Singapore receives a packet with CN $77$, which is implausible for any of the already identified devices in the area (case \ref{case:unknown} applies).
    At this point, the packet cannot yet be processed further, because the transmitter cannot be identified solely based on this one packet.
    Thirty minutes later, the same base station receives a packet with CN $9$.
    The only plausible increment for the sequence $77 \to 9$ is $21$ with $\beta - 1 = 1$ erasure, since, by \cref{eq:map_tau}, we have $\beta \equiv 21^{-1} ~ (9 - 77) \equiv 2 \pmod q$.
    Since the MAC validity checks for both packets succeed, the backend can now process and deliver the current and previous packet.

    An increment of $1$ would have required $41$ erasures ($\beta = 42$) and an increment of $91$ would have required $101$ erasures to go from $77$ to $9$.
    Both of these scenarios are extremely unlikely.
    \label{example:unknown}
\end{example}

In the following example we will use the two trivial incrementing / decrementing permutations $\sigma_\mathrm{inc}$ and $\sigma_\mathrm{dec}$ given in \cref{ex:inc_dec} for an (unrealistically small) alphabet size $q = 20$.

\newcommand{\sysstate}[5] {
    \begin{scope}[shift={#1}]
        \node at (0, 0) (#2) {};
        \draw [very thick, draw = #5, fill = black!10!white] (0, 0) circle (0.45cm);
        \node at (0,  0.2) {#3};
        \node at (0, -0.2) {#4};
        \draw (-0.35, 0) -- (0.35, 0);
    \end{scope}
}

\begin{figure}
    \centering
    \begin{tikzpicture}
        \sysstate{(0, 0)}{state1}{3}{9}{black}
        \sysstate{(2.5, 0)}{state2}{3}{7}{black}
        \sysstate{(5,  1)}{state3}{5}{7}{black}
        \sysstate{(5, -1)}{state4}{3}{5}{black}
        \sysstate{(7.5, 1)}{state5}{5}{3}{black}
        \sysstate{(7.5, -1)}{state6}{3}{3}{black}

        \draw [shorten <= 0.1cm] ($(state1) + (0, 0.2)$) -- (1, 1.2) node[anchor = west, align = left] {Last CN attributed\\to $\sigma_\mathrm{inc}$};
        \draw [shorten <= 0.1cm] ($(state1) + (0, -0.2)$) -- (1, -1.2) node[anchor = west, align = left] {Last CN attributed\\to $\sigma_\mathrm{dec}$};
        
        \draw [-latex, shorten <= 0.4cm, shorten >= 0.4cm, dashed, draw = red!50!black] (state1.east) -- (state2.west) node[midway, above] {\footnotesize $\varepsilon (1 - \varepsilon)$};
        \draw [-latex, shorten <= 0.4cm, shorten >= 0.4cm] (state2.east) -- (state3.west) node[midway, sloped, above] {\footnotesize $\varepsilon (1 - \varepsilon)$};
        \draw [-latex, shorten <= 0.4cm, shorten >= 0.4cm, dashed, draw = red!50!black] (state2.east) -- (state4.west) node[midway, sloped, above] {\footnotesize $\varepsilon (1 - \varepsilon)$};
        \draw [-latex, shorten <= 0.4cm, shorten >= 0.4cm] (state3.east) -- (state5.west) node[midway, sloped, above] {\footnotesize $\varepsilon^3 (1 - \varepsilon)$};
        \draw [-latex, shorten <= 0.4cm, shorten >= 0.4cm, dashed, draw = red!50!black] (state4.east) -- (state6.west) node[midway, sloped, above] {\footnotesize $\varepsilon (1 - \varepsilon)$};

        \node [blue!50!black] at (1.25, -2) {``7''};
        \node [blue!50!black] at (3.75, -2) {``5''};
        \node [blue!50!black] at (6.25, -2) {``3''};
    \end{tikzpicture}
    \caption{System states $\underline u$ with transition probabilities given observations (blue)}
    \label{fig:viterbi_example}
\end{figure}
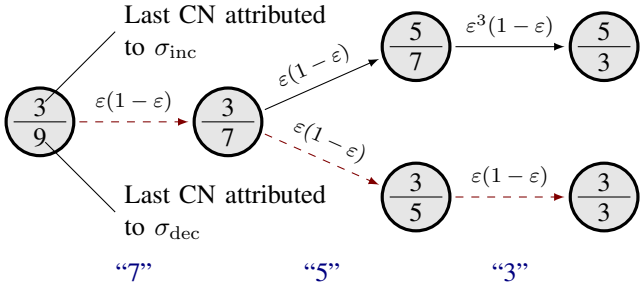

\begin{example}
    The backend receives the following sequence of CNs: $1 \to 2 \to 10 \to 9 \to 3 \to 7 \to 5 \to 3 \to 4$.
    It is likely that $1$, $2$ and $3$ originate from $\sigma_\mathrm{inc}$ and $10, 9, 7$ from $\sigma_\mathrm{dec}$.
    Upon reception of CN $v_1 = 5$, both $\sigma_\mathrm{inc}$ and $\sigma_\mathrm{dec}$ are equally likely sources of $v_1$, so the packet cannot yet be delivered and has to be processed by the HMM for identification.
    There are now two equally likely current system states: $\underline u_\mathrm{1a} = (5, 7)$ and $\underline u_\mathrm{1b} = (3, 5)$, where the CN attributed to $\sigma_\mathrm{inc}$ comes first in the tuple.

    The value $v_2 = 3$ that the backend receives just afterwards provides a hint for identification: the most likely scenario (just a single packet loss) now is that $v_2 = 3$ came from $\sigma_\mathrm{dec}$ and that the previous system state was $\underline u_\mathrm{1b} = (3, 5)$, i.e. the current system state is $u_2 = (3, 3)$.
    So, if the received $v_2$ is assumed to belong to $\sigma_\mathrm{dec}$, it is significantly more likely that $v_1 = 5$ was also transmitted by this device.
    Therefore, the value of $v_2$ that was received at a later point in time influences the backend's interpretation of $v_1$ and allows it to deliver the packet with $v_1$, though delayed.
    This example is also illustrated in \cref{fig:viterbi_example}, where the most likely path as determined by the Viterbi algorithm is dashed and marked in red.
    \Cref{fig:viterbi_example} does not show all the extremely unlikely system states and unlikely transitions.
    The subsequently received CN $4$ confirms the identification that was made based on $v_2 = 3$:
    We conclude that, almost certainly, the device with $\sigma_\mathrm{inc}$ has sent the sequence $1 \to 2 \to 3 \to 4$ and the device with $\sigma_\mathrm{dec}$ has sent $10 \to 9 \to 7 \to 5 \to 3$.
    
    \label{example:collision}
\end{example}

The kind of resolution presented in \cref{example:collision} is not always possible; in case identification is ambiguous, the better course of action can be to drop the packet than to risk erroneous identification.
Note that this ambiguity occurs only very rarely if we use a sufficiently large alphabet $Z_\mathrm{q}$.
The backend is aware of the ambiguity, and the residual probability of being unable to identify transmitters and having to discard packets becomes negligible compared to the PEC's erasure probability.

\section{Exemplary Case: Sigfox}
\label{sec:sigfox}
The French company \emph{Sigfox} \cite{sigfoxwebsite} operates a global LPWAN.
Their network protocol uses $32$-bit identifiers and $12$-bit SNs, thus allowing for up to $2^{32}$ addressable end devices and a nonce reuse after $2^{12} = 4096$ packets \cite{sigfoxspec}.
If Sigfox employed PERIDOT in a new protocol revision, they could increase the theoretical maximum number of devices, prolong nonce reuse cycles and reduce the number of bits that have to be transmitted all at the same time.

In order to achieve this, we construct a $(q, l)$-proper set based on \cref{thm:construction}.
For instance, we may want to encode the CN on a total of $38$ bits.
We assume that the occurrence of 49 or more consecutive packet losses is extremely unlikely, so we choose $l = 50$.
Since we need to have $q = p l \leq 2^{38}$ with $l \, | \, (p - 1)$ and $p$ must be prime, we choose the largest such prime, which is $p = 5\,497\,554\,151$.
Thus, we get increments for $m = p - 1 = 5\,497\,554\,150$ different devices.

As a result, compared to the current protocol, we
\begin{itemize}
    \item reduced packet size by 6 bit, resulting in a 5\% reduction in energy consumption for single-bit payload packets,
    \item increased the nonce reuse cycle by a factor of $q / 2^{12} \approx 67.1 \cdot 10^6$ making replay attacks infeasible and
    \item increased the theoretical maximum number of identifiable devices by approximately $28\%$
\end{itemize}
with negligible implementation costs on the end device side.

\section{Conclusion}
We introduced the PERIDOT scheme and described how identifiers and SNs can be completely replaced by transmitting CNs based on specially constructed permutations.
One particularly suitable construction was presented, though entirely different constructions may also be feasible.
We showed that PERIDOT can reduce uplink packet lengths and thus energy consumption, improve security and increase the theoretical maximum number of addressable devices.

The basic concept of an algorithm for the backend has been outlined, how exactly to incorporate geolocation information, behavior analysis and other metadata into state transition probabilities for the HMM remain open questions though.
Another research avenue would be to simulate a global LPWAN employing PERIDOT so that the influence of different parameters on the system and its scalability can be tested.

We found and referenced mathematical concepts that are related to PERIDOT in the theory of asymmetric error correction codes.
Still, PERIDOT remains to be classified and its relation to error correction codes needs to be better understood.

\bibliographystyle{IEEEtran}
\bibliography{IEEEabrv, bib}

\begin{thebibliography}{1}
\providecommand{\url}[1]{#1}
\csname url@samestyle\endcsname
\providecommand{\newblock}{\relax}
\providecommand{\bibinfo}[2]{#2}
\providecommand{\BIBentrySTDinterwordspacing}{\spaceskip=0pt\relax}
\providecommand{\BIBentryALTinterwordstretchfactor}{4}
\providecommand{\BIBentryALTinterwordspacing}{\spaceskip=\fontdimen2\font plus
\BIBentryALTinterwordstretchfactor\fontdimen3\font minus
  \fontdimen4\font\relax}
\providecommand{\BIBforeignlanguage}[2]{{%
\expandafter\ifx\csname l@#1\endcsname\relax
\typeout{** WARNING: IEEEtran.bst: No hyphenation pattern has been}%
\typeout{** loaded for the language `#1'. Using the pattern for}%
\typeout{** the default language instead.}%
\else
\language=\csname l@#1\endcsname
\fi
#2}}
\providecommand{\BIBdecl}{\relax}
\BIBdecl

\bibitem{sigfoxspec}
\emph{{Sigfox connected objects: Radio specifications v1.5}}, Sigfox, February
  2020.

\bibitem{loraspec}
\emph{LoRaWAN™ 1.1 Specification}, LoRa Alliance, October 2017.

\bibitem{walecki}
B.~Alspach, ``{The wonderful Walecki construction},'' \emph{Bulletin of the
  Institute of Combinatorics and its Applications}, vol.~52, pp. 7--19,
  December 2006.

\bibitem{tillson}
T.~W. Tillson, ``{A Hamiltonian decomposition of $K_{2m}^*$, $2m \geq 8$},''
  \emph{Journal of Combinatorial Theory, Series B}, vol.~29, no.~1, pp. 68--74,
  1980.

\bibitem{klove}
T.~{Kl{\o}ve}, B.~{Bose}, and N.~{Elarief}, ``Systematic, single limited
  magnitude error correcting codes for flash memories,'' \emph{IEEE
  Transactions on Information Theory}, vol.~57, no.~7, pp. 4477--4487, July
  2011.

\bibitem{magnitudefour}
\BIBentryALTinterwordspacing
D.~Xie and J.~Luo, ``Asymmetric single magnitude four error correcting codes,''
  \emph{CoRR}, vol. abs/1903.01148, 2019. [Online]. Available:
  \url{http://arxiv.org/abs/1903.01148}
\BIBentrySTDinterwordspacing

\bibitem{somecodes}
T.~Kl{\o}ve, J.~Luo, I.~Naydenova, and S.~Yari, ``Some codes correcting
  asymmetric errors of limited magnitude,'' \emph{Information Theory, IEEE
  Transactions on}, vol.~57, pp. 7459 -- 7472, 12 2011.

\bibitem{sigfoxwebsite}
``{Sigfox - The Global Communications Service Provider for the Internet of
  Things (IoT)},'' \url{https://www.sigfox.com/en}, accessed: 2020-03-20.

\end{thebibliography}

\end{document}